\documentclass[envcountsame]{llncs}

\usepackage{amsmath,amssymb}
\usepackage{graphicx}
\usepackage{caption}
\usepackage{subcaption}
\usepackage{xcolor}
\usepackage{algpseudocode}
\usepackage{algorithm}
\usepackage{mathtools}
\usepackage{lmodern}
\usepackage{microtype}
\usepackage{enumerate}
\usepackage[utf8]{inputenc}
\usepackage[T1]{fontenc}
\usepackage{multirow}
\usepackage{xspace}
\usepackage{url}
\usepackage{gnuplot-lua-tikz}

\newif\ifdraft

\draftfalse

\newcommand{\tung}[1]{#1}

\newcommand{\todo}[1]{\textcolor{red}{\bfseries\upshape #1}}

\newcommand{\tomath}[1]{\textproc{#1}}
\newcommand*\tovec[1]{\mathbf{#1}}

\newcommand{\N}{\mathbb{N}}
\newcommand{\R}{\mathbb{R}}
\newcommand{\Z}{\mathbb{Z}}
\newcommand{\Q}{\mathbb{Q}}

\newcommand{\nn}{\mathbf{n}}
\newcommand{\pp}{\mathbf{p}}
\newcommand{\qq}{\mathbf{q}}
\newcommand{\Ss}{\mathbf{s}}
\newcommand{\rr}{\mathbf{r}}
\newcommand{\xx}{\mathbf{x}}
\newcommand{\yy}{\mathbf{y}}
\newcommand{\zz}{\mathbf{z}}
\newcommand{\zero}{\mathbf{0}}

\newcommand{\Frame}{\operatorname{frame}}
\newcommand{\co}{\operatorname{\tung{conv}}}
\newcommand{\newton}{\operatorname{newton}}
\newcommand{\face}{\operatorname{face}}
\newcommand{\vx}{\operatorname{\tung{V}}}
\newcommand{\stropsat}{\psi}
\newcommand{\isvertex}{\varphi}
\newcommand{\isnegated}{\vartheta}
\newcommand{\generalstropsat}{\Psi}
\newcommand{\ispos}{\Theta}
\newcommand{\sign}{\operatorname{sign}}
\newcommand{\true}{\textproc{true}}
\newcommand{\false}{\textproc{false}}

\newcommand{\deq}{\mathrel{\dot=}}

\DeclarePairedDelimiter\abs{\lvert}{\rvert}%

\spnewtheorem{ulemma}[theorem]{Lemma}{\bfseries}{\upshape}

\begin{document}
        
\title{Subtropical Satisfiability} 



\author{Pascal Fontaine\inst{1},
  Mizuhito Ogawa\inst{2},
  Thomas~Sturm\inst{1,3},
  Xuan Tung Vu\inst{1,2}\thanks{The order of authors is strictly alphabetic.}}
\institute{
        University of Lorraine, CNRS, Inria, and LORIA, Nancy, France\\
        \email{\{Pascal.Fontaine,thomas.sturm\}@loria.fr}
        \and
        Japan Advanced Institute of Science and Technology\\
        \email{\{mizuhito,tungvx\}@jaist.ac.jp}
        \and
        MPI Informatics and Saarland University, Germany\\
        \email{sturm@mpi-inf.mpg.de}
}

\maketitle      


\begin{abstract}
  Quantifier-free nonlinear arithmetic (QF\_NRA) appears in many applications of
  satisfiability modulo theories solving (SMT). Accordingly, efficient reasoning
  for corresponding constraints in SMT theory solvers is highly relevant. We
  propose a new incomplete but efficient and terminating method to identify
  satisfiable instances. The method is derived from the subtropical method
  recently introduced in the context of symbolic computation for computing real
  zeros of single very large multivariate polynomials. Our method takes as input
  conjunctions of strict polynomial inequalities, which represent more than 40\%
  of the QF\_NRA section of the SMT-LIB library of benchmarks. The method takes
  an abstraction of polynomials as exponent vectors over the natural numbers
  tagged with the signs of the corresponding coefficients. It then uses, in
  turn, SMT to solve linear problems over the reals to heuristically find
  suitable points that translate back to satisfying points for the original
  problem. Systematic experiments on the SMT-LIB demonstrate that our method is
  not a sufficiently strong decision procedure by itself but a valuable
  heuristic to use within a portfolio of techniques.
\end{abstract}

\section{Introduction}
Satisfiability Modulo Theories (SMT) has been blooming in recent years, and many
applications rely on SMT solvers to check the satisfiability of numerous and
large formulas~\cite{Barrett14,Barrett19}. Many of those applications use
arithmetic. In fact, linear arithmetic has been one of the first theories
considered in SMT.

Several SMT solvers handle also non-linear arithmetic theories. To be precise,
some SMT solvers now support constraints of the form $p \bowtie 0$, where
$\bowtie\ \in \{=,\leq,<\}$ and $p$ is a polynomial over real or integer
variables. Various techniques are used to solve these constraints over reals,
e.g., cylindrical algebraic decomposition
(RAHD~\cite{Passmore09combineddecision,CombinedDecision}, Z3 4.3~\cite{nlsat}),
virtual substitution (SMT-RAT~\cite{smtrat}, Z3~3.1), interval constraint \pagebreak
propagation~\cite{benhamou:hal-00480814} (HySAT-II~\cite{isat},
dReal~\cite{dRealCADE13,dReal_FMCAD}, RSolver~\cite{rsolver},
RealPaver~\cite{Granvilliers06realpaver:an}, raSAT~\cite{raSAT-ijcar2016}),
CORDIC (CORD~\cite{cordic}), and linearization
(IC3-NRA-proves~\cite{Cimatti2017}). Bit-blasting
(
MiniSmt~\cite{Zankl:2010:SNR:1939141.1939168}) and linearization
(Barcelogic~\cite{Barcelogic08}) can be used for integers.

We present here an incomplete but efficient method to detect the satisfiability
of large conjunctions of constraints of the form $p > 0$ where $p$ is a
multivariate polynomial with strictly positive real variables. The method
quickly states that the conjunction is satisfiable, or quickly returns unknown.
Although seemingly restrictive, 40\% of the quantifier-free non-linear real
arithmetic (QF\_NRA) category of the SMT-LIB is easily reducible to the
considered fragment. Our method builds on a \emph{subtropical} technique that
has been found effective to find roots of very large polynomials stemming from
chemistry and systems biology~\cite{subtropical,ErramiEiswirth:15a}. Recall that
a univariate polynomial with a positive head coefficient diverges positively as
$x$ increases to infinity. Intuitively, the subtropical approach generalizes
this observation to the multivariate case and thus to higher dimensions.

In Sect.~\ref{SE:basic} we recall some basic definitions and facts. In
Sect.~\ref{SE:subtropical} we provide a short presentation of the original
method~\cite{subtropical} and give some new insights for its foundations. In
Sect.~\ref{SE:multiplepol}, we extend the method to multiple polynomial
constraints. We then show in Sect.~\ref{SE:gensol} that satisfiability modulo
linear theory is particularly adequate to check for applicability of the method.
In Sect.~\ref{SE:experiments}, we provide experimental evidence that the method
is suited as a heuristic to be used in combination with other, complete,
decision procedures for non-linear arithmetic in SMT. It turns out that our
method is quite fast at either detecting satisfiability or failing. In
particular, it finds solutions for problems where state-of-the-art non-linear
arithmetic SMT solvers time out. Finally, in Sect.~\ref{SE:conclusions}, we
summarize our contributions and results, and point at possible future research
directions.

\section{Basic Facts and Definitions}\label{SE:basic}

For $a \in \mathbb{R}$, a vector $\xx = (x_1, \dots, x_d)$ of variables, and
$\pp = (p_1, \dots, p_d) \in \mathbb{R}^d$ we use notations
$a^\pp = (a^{p_1}, \dots, a^{p_d})$ and
$\xx^\pp = (x_1^{p_1}, \dots, x_d^{p_d})$. The \emph{frame} $F$ of a
multivariate polynomial $f\in\Z[x_1,\dots,x_d]$ in sparse distributive
representation
\begin{displaymath}
  f=\sum_{\pp \in F} f_\pp \tung{\xx}^\pp, \quad \tung{f_\pp \neq 0},\quad 
  F\subset\N^d,
\end{displaymath}
is uniquely determined, and written $\Frame(f)$.  It can be partitioned into a
positive and a negative frame, according to the sign of $f_\pp$:
\begin{displaymath}
  \Frame^+(f) = \{\,\pp\in\Frame(f)\mid f_\pp>0\,\},\ \ \
  \Frame^-(f) = \{\,\pp\in\Frame(f)\mid f_\pp<0\,\}.
\end{displaymath}

For $\pp$, $\qq\in\R^d$ we define
$\overline{\pp\qq}=\{\,\lambda\pp+(1-\lambda)\qq\in\R^n\mid\lambda\in[0,1]\,\}$.
Recall that $S\subseteq \mathbb{R}^d$ is \emph{convex} if
$\overline{\pp\qq}\subseteq S$ for all $\pp$, $\qq\in S$. Furthermore, given any
$S\subseteq\R^d$, the \emph{convex hull} $\co(S)\subseteq\R^d$ is the unique
inclusion-minimal convex set containing $S$.
\begin{figure}[t]
        \captionsetup[subfigure]{position=b}
        \centering
        \begin{subfigure}[b]{.68\columnwidth}
                \centering
                \subcaptionbox{The frame and the Newton polytope $P$ of 
                        $f$\label{fig:subtropical-single}}
                {\includegraphics[width=\linewidth]{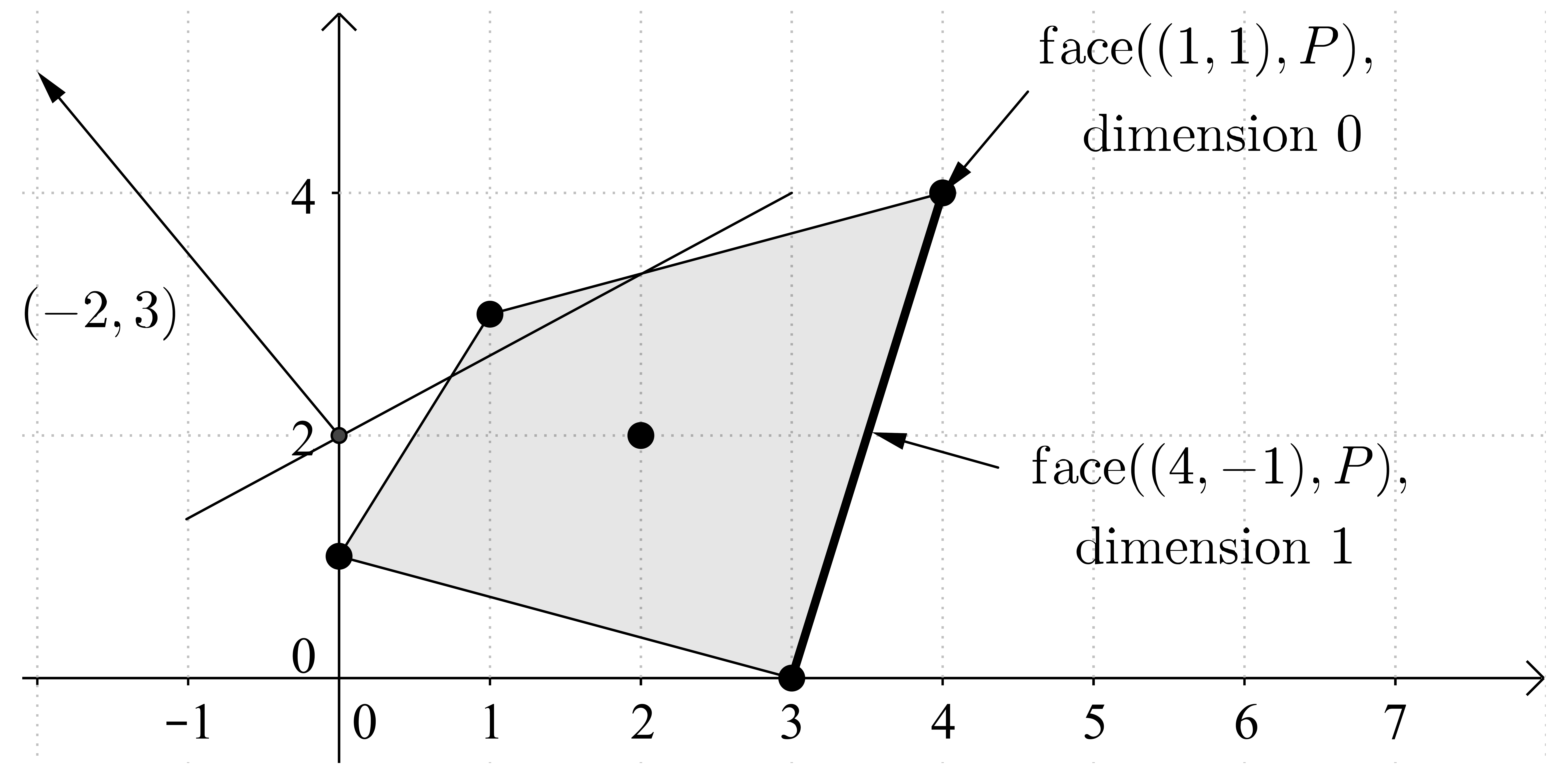}}
        \end{subfigure}%
        \begin{subfigure}[b]{.32\columnwidth}
                \centering
                \subcaptionbox{The variety of $f$ and the moment curve
                        $(a^{-2},a^3)$
                        \label{fig:subtropical-single-graph}}
                {\includegraphics[width=\linewidth]{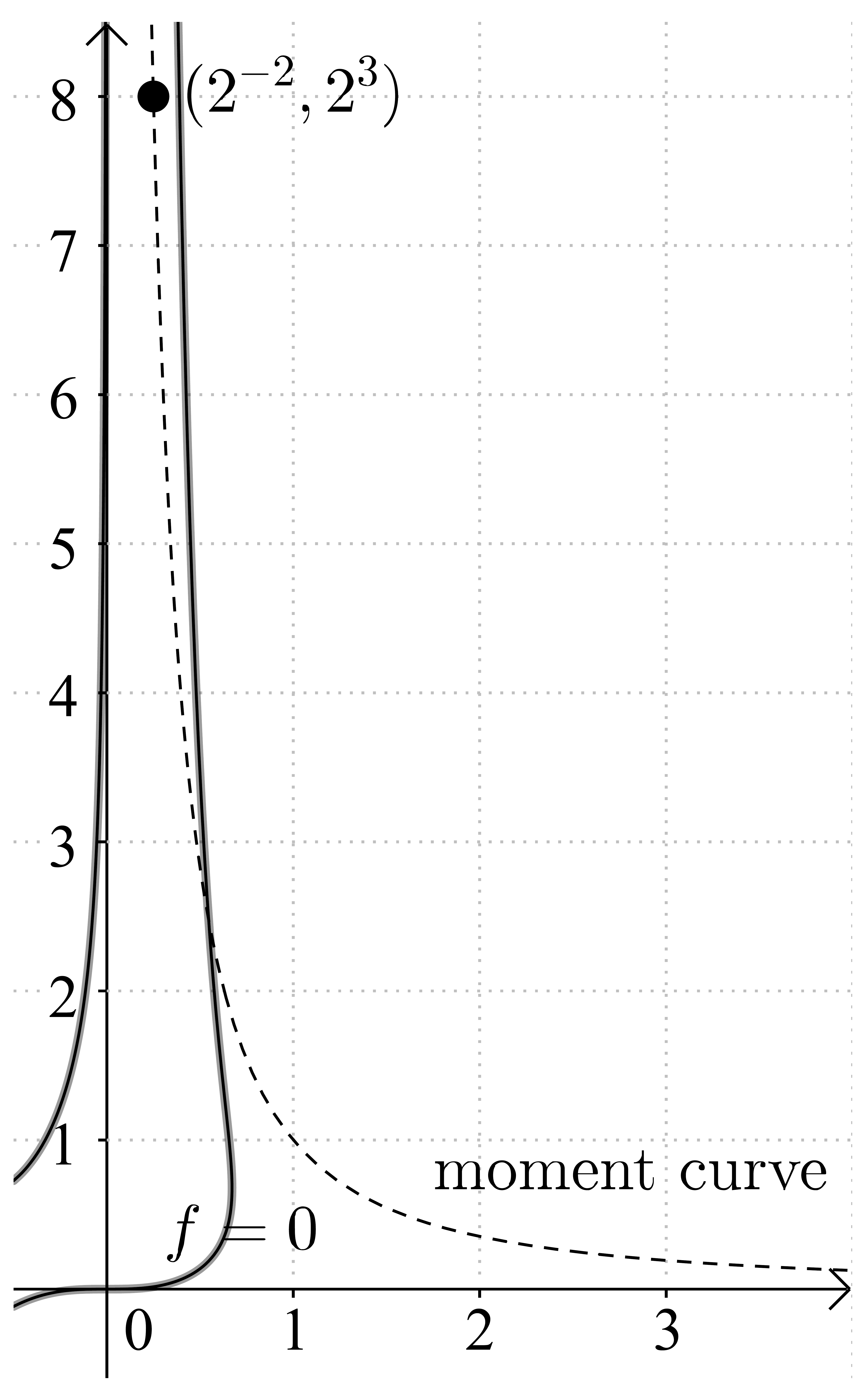}}
        \end{subfigure}
        \caption{An illustration of Example~\ref{eg:subtropical-single}, where $f=y +2xy^3 -3x^2y^2- x^3 - 
                4x^4y^4$}
        \label{fig:test}
\end{figure}
The \emph{Newton polytope} of a polynomial $f$ is the convex hull of its frame,
$\newton(f)=\co(\Frame(f))$. Fig.~\ref{fig:subtropical-single} illustrates the
Newton polytope of
\begin{displaymath}
y +2xy^3 -3x^2y^2 - x^3 - 4x^4y^4\in\Z[x,y],
\end{displaymath}
which is the convex hull of its frame
$\{(0, 1), (1, 3),(2, 2), (3, 0), (4, 4)\}\subset\N^2$. As a convex hull of a
finite set of points, the Newton polytope is bounded and thus indeed a
polytope~\cite{lineartheory}.

The \emph{face}~\cite{lineartheory} of a polytope $P\subseteq\R^d$ with respect 
to a
vector $\nn\in\R^d$ is
\begin{displaymath}
  \face(\nn, P) = \{\,\pp\in P \mid  \text{$\nn^T\pp \geq \nn^T\qq$ for all $\qq  \in P$}\, \}.
\end{displaymath}
Faces of dimension $0$ are called \emph{vertices}. We denote by $\vx(P)$ the set
of all vertices of $P$. We have $\pp\in\vx(P)$ if and only if there exists
$\nn\in\R^d$ such that $\nn^T\pp > \nn^T\qq$ for all $\qq\in P\setminus\{\pp\}$.
In Fig.\ref{fig:subtropical-single}, $(4,4)$ is a vertex of the Newton polytope
with respect to $(1,1)$.

It is easy to see that for finite $S\subset\R^d$ we have
\begin{equation}
  \vx(\co(S))\subseteq S\subseteq\co(S).\label{eq:1}
\end{equation}
The following lemma gives a
characterization of $\vx(\co(S))$:
\begin{lemma}
  \label{lem:vertex-hyperplane}
  Let $S\subset\R^d$ be finite, and let $\pp\in S$. The following are
  equivalent:
  \begin{enumerate}[(i)]
  \item $\pp$ is a vertex of $\co(S)$ with respect to $\nn$.
  \item There exists a hyperplane $H:\nn^T\xx + c = 0$ that strictly separates
    $\pp$ from $S \setminus \{\pp\}$, and the normal vector $\nn$ is directed
    from $H$ towards $\pp$.
  \end{enumerate}
\end{lemma}

\begin{proof}
  Assume (i). Then there exists $\nn\in\R^d$ such that $\nn^T\pp > \nn^T\qq$ for
  all $\qq\in S\setminus\{\pp\}\subseteq\co(S)\setminus\{\pp\}$. Choose
  $\qq_0\in S\setminus\{\pp\}$ such that $\nn^T\qq_0$ is maximal, and choose $c$
  such that $\nn^T\pp>-c>\nn^T\qq_0$. Then $\nn^T\pp+c>0$ and
  $\nn^T\qq+c\leq\nn^T\qq_0+c<0$ for all $\qq\in S\setminus\{\pp\}$. Hence
  $H:\nn^T\pp+c=0$ is the desired hyperplane.
  
  Assume (ii). It follows that $\nn^T\pp+c>0>\nn^T\qq+c$ for all
  $\qq\in S\setminus\{\pp\}$. If $\qq\in S\setminus\{\pp\}$, then
  $\nn^T\pp>\nn^T\qq$. If, in contrast,
  $\qq\in(\co(S)\setminus S)\setminus\{\pp\}=\co(S)\setminus S$, then
  $\qq=\sum_{\Ss\in S} t_\Ss \Ss$, where $t_\Ss\in[0,1]$,
  $\sum_{\Ss\in S}t_\Ss=1$, and at least two $t_\Ss$ are greater than $0$. It
  follows that
  \begin{displaymath}
    \nn^T\qq=\nn^T\sum_{\Ss\in S} t_\Ss \Ss<\nn^T\pp\sum_{\Ss\in S}
    t_\Ss=\nn^T\pp.\quad\qed
  \end{displaymath}
\end{proof}

\newcommand{\vbasis}{vertex cluster\xspace}

Let $S_1$, \dots,~$S_m \subseteq \mathbb{R}^d$, and let $\nn\in\R^d$. If there
exist $\pp_1\in S_1$, \dots, $\pp_n\in S_m$ such that each $\pp_i$ is a vertex
of $\co(S_i)$ with respect to $\nn$, then the (unique) \emph{\vbasis} of
$\{S_i\}_{i\in\{1,\dots,m\}}$ with respect to $\nn$ is defined as
$(\pp_1,\dots,\pp_m)$.

\section{Subtropical Real Root Finding Revisited}\label{SE:subtropical}
This section improves on the \tung{original method} described in
\cite{subtropical}. It furthermore lays some theoretical foundations to better
understand the limitations of the heuristic approach.
The method finds real zeros with all positive coordinates of a multivariate
polynomial $f$ in three steps:
\begin{enumerate}
\item Evaluate $f(1,\dots,1)$. If this is $0$, we are done. If this is greater
  than $0$, then consider $-f$ instead of $f$. We may now assume that we have
  found $f(1,\dots,1)<0$.
\item Find $\pp$ with all positive coordinates such that $f(\pp)>0$.
\item Use the Intermediate Value Theorem (a continuous function with positive
  and negative values has a zero) to construct a root of $f$ on the line segment
  $\overline{\mathbf{1}\pp}$.
\end{enumerate}
We focus here on Step 2. Our technique builds on \cite[Lemma 4]{subtropical},
which we are going to restate now in a slightly generalized form. While the
original lemma required that $\pp\in\Frame(f)\setminus\{\zero\}$, inspection of the
proof shows that this limitation is not necessary:
\begin{lemma}
  \label{lem:base-stropsat}
  Let $f$ be a polynomial, and let $\pp \in \Frame(f)$ be a vertex of
  $\newton(f)$ with respect to $\nn \in \mathbb{R}^d$. Then there exists
  $a_0 \in \mathbb{R}^+$ such that for all $a \in \mathbb{R}^+$ with
  $a \ge a_0$ the following holds:
  \begin{enumerate}
  \item
    $\abs{f_\pp\ a^{\nn^T\pp}} >
    \abs{\sum_{\qq \in \Frame(f)\setminus
        \{\pp\}}f_\qq\ a^{\nn^T\qq}}$,
  \item $\sign(f(a^\nn)) = 
  \sign(f_\pp)$.\qed
  \end{enumerate}
\end{lemma}
In order to find a point with all positive coordinates where $f > 0$, the
original method iteratively examines each
$\pp \in \Frame^+(f) \setminus \{\tovec{0}\}$ to check if it is a vertex
of $\newton(f)$ with respect to some $\nn \in \mathbb{R}^d$. In the
positive case, Lemma~\ref{lem:base-stropsat} guarantees for large enough
$a \in \mathbb{R}^+$ that
$\sign(f(a^{\nn})) = \sign(f_\pp)=1$, in
other words, $f(a^{\nn}) > 0$.

\begin{example}
\label{eg:subtropical-single}
Consider $f = y +2xy^3 -3x^2y^2 - x^3 - 4x^4y^4$.
Figure~\ref{fig:subtropical-single} illustrates the frame and the Newton
polytope of $f$, of which $(1,3)$ is a vertex with respect to $(-2, 3)$.
Lemma~\ref{lem:base-stropsat} ensures that $f(a^{-2}, a^{3})$ is strictly
positive for sufficiently large positive $a$. For example,
$f(2^{-2}, 2^{3}) = \frac{51193}{256}$.
Figure~\ref{fig:subtropical-single-graph} shows how the moment curve
$(a^{-2}, a^{3})$ with $a \geq 2$ will not leave the sign invariant region of
$f$ that contains $(2^{-2}, 2^{3})$.
\end{example}
An exponent vector $\zero\in\Frame(f)$ corresponds to an absolute summand
$f_\zero$ in $f$. Its above-mentioned explicit exclusion in \cite[Lemma
4]{subtropical} originated from the false intuition that one cannot achieve
$\sign(f(a^\nn))=\sign(f_\zero)$ because the monomial $f_\zero$ is invariant
under the choice of $a$. However, inclusion of $\zero$ can yield a normal vector
$\nn$ which renders all other monomials small enough for $f_\zero$ to dominate.

Given a finite set $S \subset \mathbb{R}^d$ and a point $\pp \in S$, the
original method uses linear programming to determine if $\pp$ is a vertex
of $\co(S)$ w.r.t. some vector $\nn \in \mathbb{R}^d$.
Indeed, from Lemma~\ref{lem:vertex-hyperplane}, the problem can be reduced to
finding a hyperplane $H: \nn^T\xx+c = 0$ that strictly separates
$\pp$ from $S \setminus \{\pp\}$ with the normal vector $\nn$
pointing from $H$ to $\pp$. This is equivalent to solving the following
linear problem with $d + 1$ real variables $\nn$ and $c$:
\begin{equation}
\label{eq:isvertex}
\isvertex(\pp, S, \nn, c) \deq  
\nn^T\pp+c > 0 \land \bigwedge_{\qq \in S \setminus 
\{\pp\}} \nn^T\qq+c < 0.
\end{equation}

Notice that with the occurrence of a nonzero absolute summand the corresponding
point $\tovec{0}$ is generally a vertex of the Newton polytope with respect to
$\tovec{-1} = (-1, \dots, -1)$.
This raises the question whether there are other special points that are
certainly vertices of the Newton polytope. In fact, $\tovec{0}$ is a
lexicographic minimum in $\Frame(f)$, and it is not hard to see that minima and
maxima with respect to lexicographic orderings are generally vertices of the
Newton polytope.

We are now going to generalize that observation. A \emph{monotonic total preorder}
${\preceq}\subseteq\Z^d\times\Z^d$ is defined as follows:
\begin{enumerate}[(i)]
\item $\xx\preceq\xx$ (reflexivity)
\item $\xx\preceq\yy\land\yy\preceq\zz\longrightarrow\xx\preceq\zz$
  (transitivity)
\item $\xx\preceq\yy\longrightarrow\xx+\zz\preceq\yy+\zz$ (monotonicity)
\item $\xx\preceq\yy\lor\yy\preceq\xx$ (totality).
\end{enumerate}
The difference to a total order is the missing anti-symmetry. As an example in
$\Z^2$ consider $(x_1,x_2)\preceq(y_1,y_2)$ if and only if
$x_1+x_2\leq y_1+y_2$.
Then $-2\preceq 2$ and $2\preceq -2$ but $-2\neq 2$. Our
definition of $\preceq$ on the extended domain $\Z^d$ guarantees a cancellation
law $\xx+\zz\preceq \yy+\zz\longrightarrow \xx\preceq \yy$ also on $\N^d$. The
following lemma follows by induction using monotonicity and cancellation:

\begin{lemma}
  For $n\in\N\setminus\{0\}$ denote as usual the $n$-fold addition of $\xx$ as
  $n\odot\xx$. Then
  $\xx\preceq \yy\longleftrightarrow n\odot\xx\preceq n\odot\yy$.\qed
\end{lemma}
%
%
%
Any monotonic preorder $\preceq$ on $\Z^d$ can be extended to $\Q^d$: Using a
suitable principle denominator $n\in\N\setminus\{0\}$ define
\begin{displaymath}
  \left(\frac{x_1}{n},\dots,\frac{x_d}{n}\right)\preceq
  \left(\frac{y_1}{n},\dots,\frac{y_d}{n}\right)
  \quad\text{if and only if}\quad
  (x_1,\dots,x_d)\preceq
  (y_1,\dots,y_d).
\end{displaymath}
This is well-defined.



Given $\xx\preceq\yy$ we have either $\yy\npreceq\xx$ or $\yy\preceq\xx$. In the
former case we say that $\xx$ and $\yy$ are \emph{strictly} preordered and write
$\xx\prec\yy$. In the latter case they are \emph{not} strictly preordered, i.e.,
$\xx\nprec\yy$ although we might have $\xx\neq\yy$. In particular, reflexivity
yields $\xx\preceq\xx$ and hence certainly $\xx\nprec\xx$.

\begin{example}
  Lexicographic orders are monotonic total orders and thus monotonic total
  preorders. Hence our notion covers our discussion of the absolute summand
  above. Here are some further examples: For $i\in\{1,\dots,d\}$ we define
  $\xx\preceq_i\yy$ if and only if $\pi_i(\xx)\leq\pi_i(\yy)$, where $\pi_i$
  denotes the $i$-th projection. Similarly, $\xx\succeq_{i}\yy$ if and only if
  $\pi_i(\xx)\geq\pi_i(\yy)$. Next, $\xx\preceq_\Sigma\yy$ if and only if
  $\sum_ix_i\leq\sum_iy_i$. Our last example is going to be
  instrumental with the proof of the next theorem: Fix $\nn \in \mathbb{R}^d$,
  and define for $\pp$, $\pp'\in\Z^d$ that $\pp \preceq_\nn \tovec{p'}$ if and
  only if $\nn^T\pp\leq \nn^T\tovec{p'}$.
\end{example}

\begin{theorem}\label{th:admissible}
  Let $f\in\Z[x_1,\dots,x_d]$, and let $\pp\in\Frame(f)$. Then the following are
  equivalent:
  \begin{enumerate}[(i)]
  \item $\pp\in\vx(\newton(f))$
  \item There exists a  monotonic total preorder $\preceq$ on $\Z^d$ such that 
    \begin{displaymath}
      \pp=\max\nolimits_\prec(\Frame(f)).
\end{displaymath}
  \end{enumerate}
\end{theorem}

\begin{proof}
  Let $\pp$ be a vertex of $\newton(f)$ specifically with
  respect to $\nn$. By our definition of a vertex in Sect.~\ref{SE:basic}, $\pp$
  is the maximum of $\Frame(f)$ with respect to $\prec_\nn$.
  
  Let, vice versa, $\preceq$ be a monotonic total preorder on $\Z^d$, and let
  $\pp=\max_\prec(\Frame(f))$. Shortly denote $V=\vx(\newton(f))$, and assume
  for a contradiction that $\pp\notin V$. Since
  $\pp\in\Frame(f)\subseteq\newton(f)$, we have
  \begin{displaymath}
    \pp = \sum_{\Ss \in V}t_\Ss \Ss,\quad\text{where}\quad
    t_\Ss \in [0, 1]\quad\text{and}\quad
    \sum_{\Ss \in V}t_\Ss=1.
  \end{displaymath}
  According to (\ref{eq:1}) in Sect.~\ref{SE:basic} we know that
  $V\subseteq\Frame(f)\subseteq\newton(f)$. It follows that $\Ss\prec\pp$ for
  all $\Ss\in V$, and using monotony we obtain
  \begin{displaymath}
    \pp\prec\sum_{\Ss \in V}t_\Ss \pp=\left(\sum_{\Ss \in V}t_\Ss\right)\pp=\pp.
  \end{displaymath}
  On the other hand, we know that generally $\pp\nprec\pp$, a contradiction.\qed
\end{proof}

In Fig.~\ref{fig:subtropical-single} we have
$(0,1)=\max_{\succeq_1}(\Frame(f))$, $(3,0)=\max_{\succeq_2}(\Frame(f))$, and
$(4,4)=\max_{\preceq_1}(\Frame(f))=\max_{\preceq_2}(\Frame(f))$. This shows
that, besides contributing to our theoretical understanding, the theorem can be
used to substantiate the efficient treatment of certain special cases in
combination with other methods for identifying vertices of the Newton polytope.




\begin{corollary}
  Let $f\in\Z[x_1,\dots,x_d]$, and let $\pp\in\Frame(f)$. If $p=\max(\Frame(f))$
  or $p=\min(\Frame(f))$ with respect to an admissible term order in the sense
  of Gröbner Basis theory \cite{Buchberger:65a}, then $p\in\vx(\newton(f))$.\qed
\end{corollary}

It is one of our research goals to identify and characterize those polynomials
where the subtropical heuristic succeeds in finding positive points. We are now
going to give a necessary criterion. Let $f\in\Z[x_1,\dots,x_d]$, define
${\Pi(f)=\{\,\rr\in\mathopen]0,\infty\mathclose[^d\mid f(\rr)>0\,\}}$, and
denote by $\overline{\Pi(f)}$ its closure with respect to the natural topology.
In Lemma~\ref{lem:base-stropsat}, when $a$ tends to $\infty$, $a^{\nn}$ will
tend to some $\rr \in \{0, \infty\}^d$. If $\rr=\zero$, then
$\zero\in\overline{\Pi(f)}$. Otherwise, $\Pi(f)$ is unbounded. Consequently, for
the method to succeed, $\Pi$ must have at least one of those two properties.
Figure~\ref{fig:characterization-cases} illustrates four scenarios: the
subtropical method succeeds in the first three cases while it fails to find a
point in $\Pi(f)$ in the last one. The first sub-figure presents a case where
$\Pi(f)$ is unbounded. The second and third sub-figures illustrate cases where
the closure of $\Pi(f)$ contains $(0,0)$. In the fourth sub-figure where neither
$\Pi(f)$ is unbounded nor its closure contains $(0,0)$, the method cannot find
any positive value of the variables for $f$ to be positive.

\begin{figure}[t]
       \centering
       \begin{subfigure}{.4\columnwidth}
               \centering
               \includegraphics[width=\linewidth]{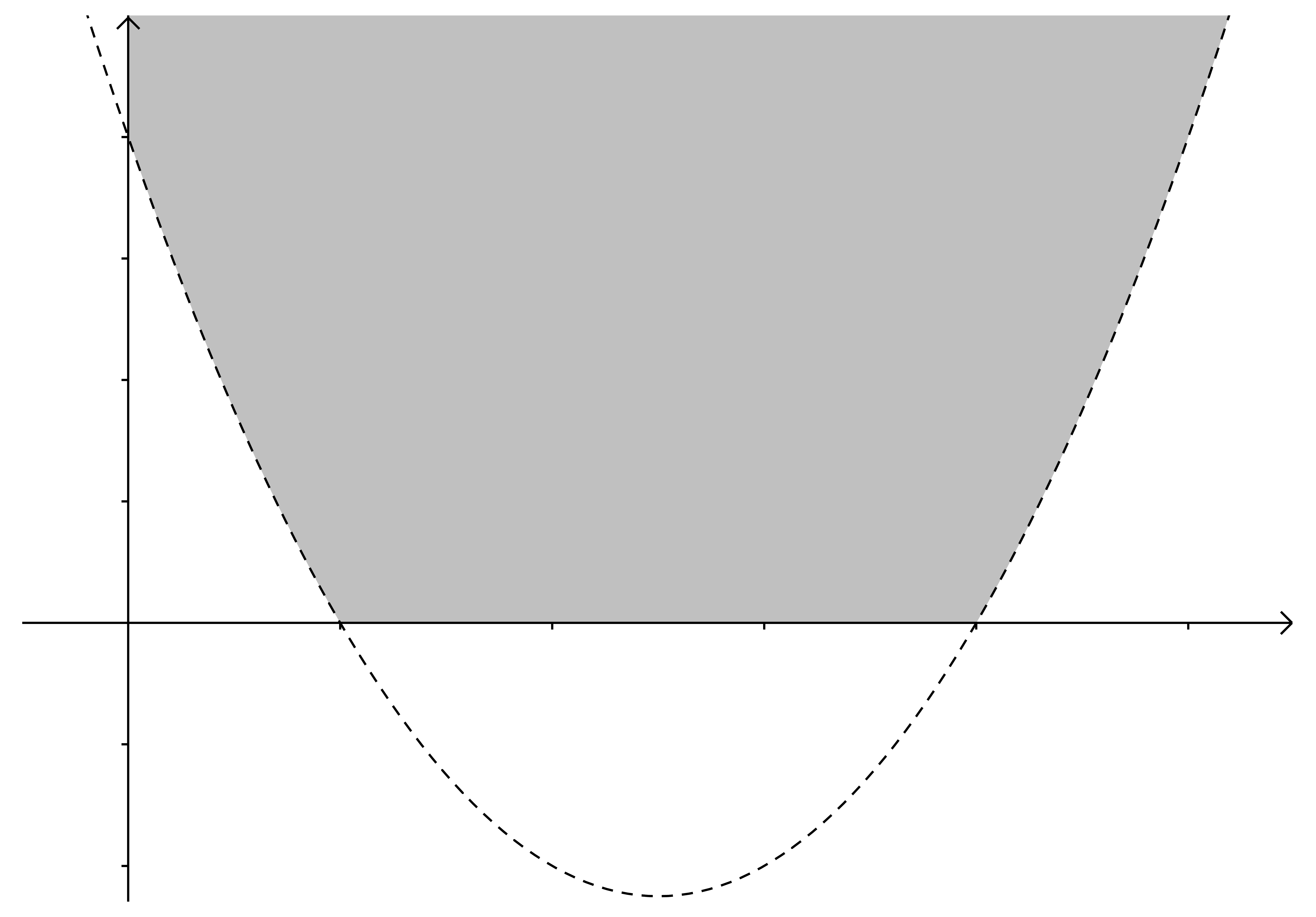}
               \caption{$f=y-x^2+5x-4$}
       \end{subfigure}%
       \quad\begin{subfigure}{.4\columnwidth}
               \centering
               \includegraphics[width=\linewidth]{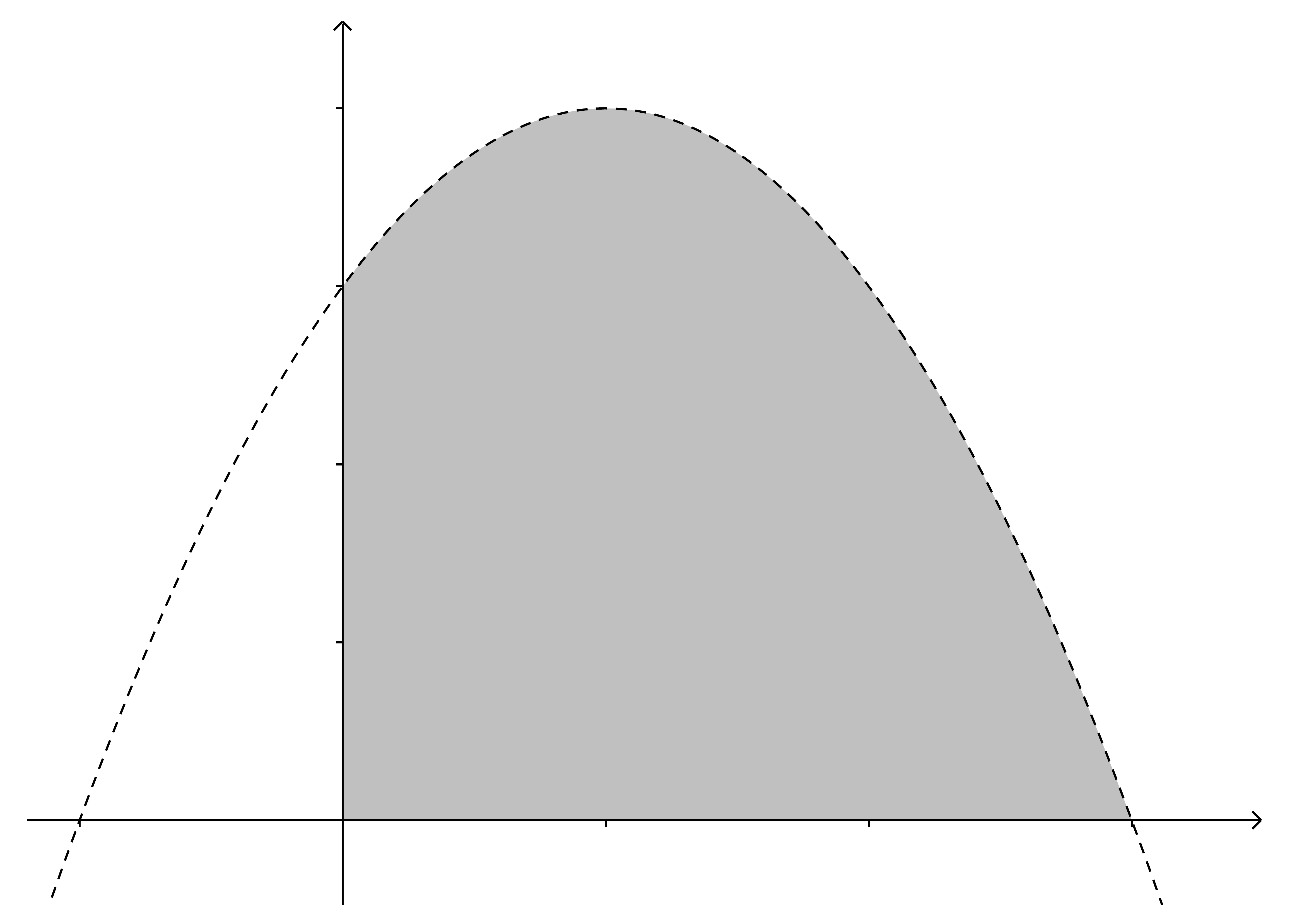}
               \caption{$f=-y-x^2+2x+3$}
       \end{subfigure}\\[1ex]
                \begin{subfigure}{.4\columnwidth}
                        \centering
                        \includegraphics[width=\linewidth]{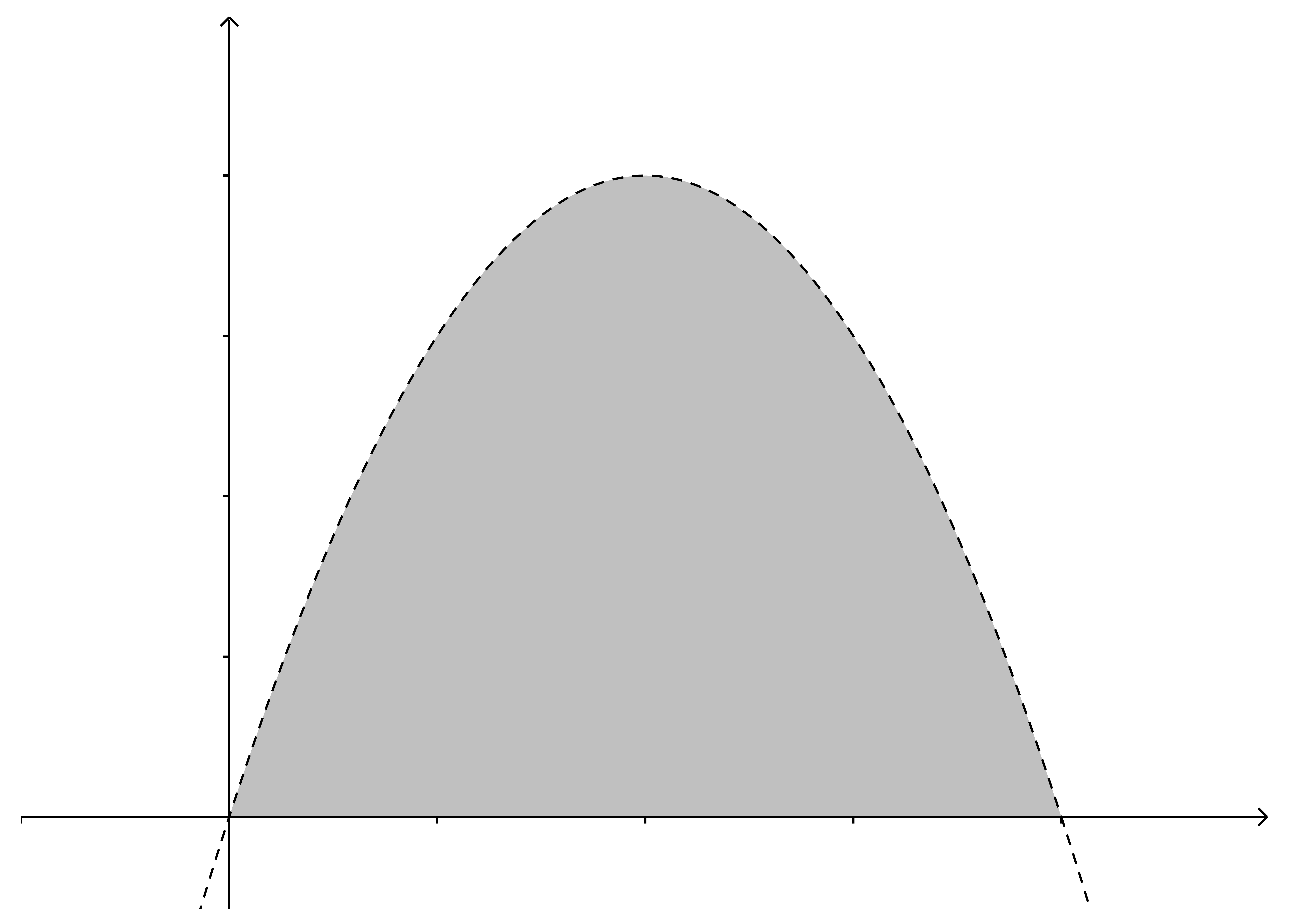}
                        \caption{$f=-y-x^2+4x$}
                \end{subfigure}
                \quad\begin{subfigure}{.4\columnwidth}
                        \centering
                        \includegraphics[width=\linewidth]{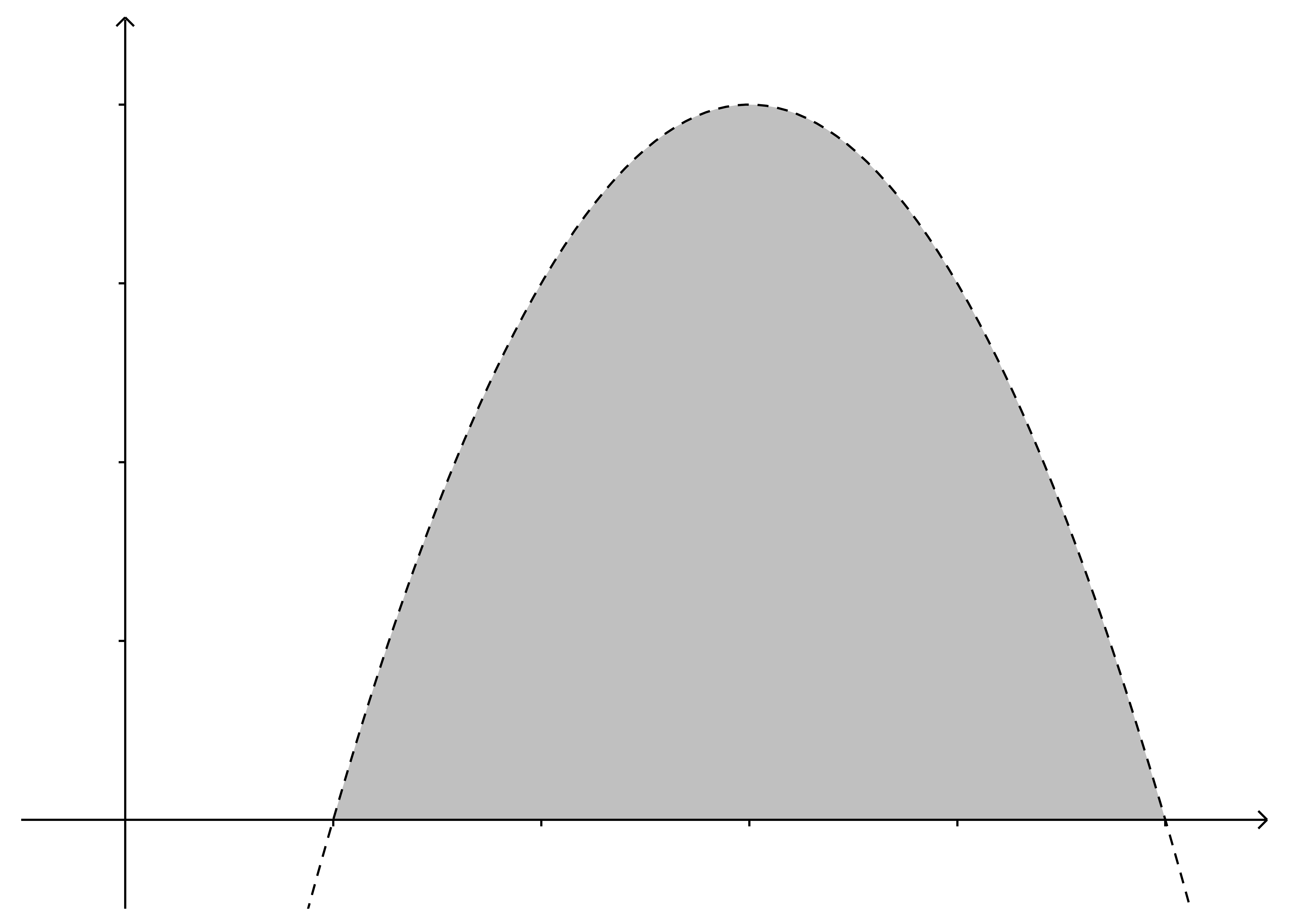}
                        \caption{$f=-y-x^2+6x-5$}
                \end{subfigure}
                \caption{Four scenarios of polynomials for the subtropical
                  method. The shaded regions show $\Pi(f)$.}
        \label{fig:characterization-cases}
\end{figure}

\section{Positive Values of Several Polynomials}\label{SE:multiplepol}

The subtropical method as presented in \cite{subtropical} finds zeros with all
positive coordinates of one single multivariate polynomial. This requires to
find a corresponding point with a positive value of the polynomial. In the
sequel we restrict ourselves to this sub-task. This will allow us generalize
from one polynomial to simultaneous positive values of finitely many
polynomials.

\subsection{A Sufficient Condition}

With a single polynomial, the existence of a positive vertex of the Newton
polytope guarantees the existence of positive real choices for the variables
with a positive value of that polynomial. For several polynomials we introduce a
more general notion: A sequence $(\pp_1, \dots, \pp_m)$ is a \emph{positive
  \vbasis} of $\{f_i\}_{i\in\{1,\dots, m\}}$ with respect to
$\nn \in \mathbb{R}^d$ if it is a \vbasis of
$\{\Frame(f_i)\}_{i \in \{1,\dots,m\}}$ with respect to $\nn$ and
$\pp_i \in \Frame^+(f_i)$ for all $i \in \{1, \dots, m\}$. The existence of a
positive \vbasis will guarantee the existence of positive real choices of the
variables such that all polynomials $f_1$, \dots,~$f_m$ are simultaneously
positive. The following lemma is a corresponding generalization of
Lemma~\ref{lem:base-stropsat}:
\begin{lemma}\label{lem:multiple-polys}
  If there exists a \vbasis $(\pp_1,\dots, \pp_m)$ of
  $\{\Frame(f_i)\}_{i \in \{1,\dots,m\}}$ with respect to
  $\nn \in \mathbb{R}^n$, then there exists $a_0 \in \mathbb{R^+}$ such that the
  following holds for all $a \in \mathbb{R}^+$ with $a \ge a_0$ and all
  $i \in \{1, \dots, m\}$:
\begin{enumerate}
\item
  $\abs{(f_i)_{\pp_i}\ a^{\nn^T\pp_i}} >
  \abs{\sum_{\qq \in \Frame(f_i)\setminus
      \{\pp_i\}}(f_i)_\qq\  a^{\nn^T\qq}}$,
\item $\sign(f_i(a^\nn)) = 
\sign((f_i)_{\pp_i})$.
\end{enumerate}
\end{lemma}

\begin{proof}
From~\cite[Lemma 4]{subtropical}, for each $i\in\{1,\dots, m\}$, there exist 
$a_{0,i} \in \R^+$ such that for all $a \in \R^+$ with $a \geq a_{0,i}$ 
the following holds:
\begin{enumerate}
\item
  $\abs{(f_i)_{\pp_i}\ a^{\nn^T\pp_i}} >
  \abs{\sum_{\qq \in \Frame(f_i)\setminus
      \{\pp_i\}}(f_i)_{\qq}\ a^{\nn^T\qq}}$,
        \item $\sign(f_i(a^n)) = 
        \sign((f_i)_{\pp_i})$.
\end{enumerate}
It now suffices to take $a_0 = \max\{ a_{0,i}\ |\ 1 \leq i \leq m \}$.
\qed
\end{proof}
Similarly to the case of one polynomial, the following Proposition provides a
sufficient condition for the existence of a common point with positive value for
multiple polynomials.
\begin{proposition}
        \label{prop:multiple-inequalities}
                If there exists a positive 
        \vbasis $(\pp_1,\dots, \pp_m)$ of the polynomials
        $\{f_i\}_{i\in\{1,\dots, m\}}$ with respect to a vector $\nn \in 
        \mathbb{R}^d$,
        then there exists $a_0 \in \mathbb{R}^+$ such that for all $a \in 
        \mathbb{R}^+$ with $a \ge a_0$ the following holds:$$\bigwedge_{i=1}^m 
        f_i(a^\nn) > 0.$$
\end{proposition}

\begin{proof}
  For $i\in\{1,\dots,m\}$, since $\pp_i \in \Frame^+(f_i)$,
  Lemma~\ref{lem:multiple-polys} implies $f_i(a^\nn) > 0$.\qed
\end{proof}

\begin{example}
  Consider $f_1=2-xy^2z+x^2yz^3, f_2=3-xy^2z^4-x^2z-x^4y^3z^3,$ and
  $f_3 = 4 - z - y - x + 4$. The exponent vector $\tovec{0}$ is a vertex of
  $\text{newton}(f_1)$, $\newton(f_2)$, and $\newton(f_3)$ with respect to
  $(-1, -1, -1)$. Choose $a_0=2 \in \mathbb{R}^+$. Then for all
  $a \in \mathbb{R}$ with $a \ge a_0$ we have
  $f_1(a^{-1},a^{-1},a^{-1})>0 \land f_2(a^{-1},a^{-1},a^{-1}) > 0 \land
  f_3(a^{-1},a^{-1},a^{-1}) > 0$. \qed
\end{example}

\subsection{Existence of Positive Vertex Clusters}
Given polynomials $f_1$, \dots,~$f_m$,
Proposition~\ref{prop:multiple-inequalities} provides a sufficient condition,
i.e. the existence of a positive \vbasis of $\{f_i\}_{i\in\{1,\dots, m\}}$, for
the satisfiability of $\bigwedge_{i=1}^{m}f_i>0$. A straightforward method to
decide the existence of such a cluster is to verify whether each
$(\pp_1, \dots, \pp_m) \in \Frame^+(f_1) \times \cdots \times \Frame^+(f_m)$ is
a positive \vbasis by checking the satisfiability of the formula
\begin{displaymath}
\bigwedge\limits_{i \in \{1, \dots, m\}}\isvertex(\pp_i, \Frame(f_i), \nn,
c_i),
\end{displaymath}
where $\isvertex$ is defined as in (\ref{eq:isvertex}) on
p.\pageref{eq:isvertex}. This is a linear problem with $d+m$ variables $\nn$,
$c_1$, \dots,~$c_m$. Since $\Frame(f_1)$, \dots,~$\Frame(f_m)$ are finite,
checking all $m$-tuples $(\pp_1, \dots, \pp_m)$ will terminate, provided we rely
on a complete algorithm for linear programming, such as the Simplex
algorithm~\cite{simplex}, the ellipsoid method~\cite{ellipsoid}, or the interior
point method~\cite{interior-point-method}. This provides a decision procedure
for the existence of a positive \vbasis of $\{f_i\}_{i\in\{1,\dots, m\}}$.
However, this requires checking all candidates in
$\Frame^+(f_1) \times \cdots \times \Frame^+(f_m)$.

We propose
to use instead state-of-the-art SMT solving techniques over linear real
arithmetic to examine whether or not $\{f_i\}_{i\in\{1,\dots, m\}}$ has a
positive \vbasis with respect to some $\nn \in \mathbb{R}^d$. In the positive
case, a solution for $\bigwedge_{i=1}^{m}f_i>0$ can be constructed as $a^\nn$
with a sufficiently large $a \in \mathbb{R}^+$.

To start with, we provide a characterization for the positive frame of a single
polynomial to contain a vertex of the Newton polytope.
\begin{lemma}\label{lemma:gurobi-lra-eq}
  Let $f\in\Z[\xx]$. The following are equivalent:
  \begin{enumerate}[(i)]
  \item There exists a vertex $\pp \in \Frame^+(f)$ of
    $\newton(f)= \co(\Frame(f))$ with respect to $\nn\in \mathbb{R}^d$.
  \item There exists a vertex $\pp' \in \Frame^+(f)$ such that $\pp'$ is also a vertex of
    ${\co(\Frame^-(f) \cup \{\pp'\})}$ with respect to
    $\tovec{n'}\in \mathbb{R}^d$.
  \end{enumerate}
\end{lemma}

\begin{proof}
  Assume (i). Take $\pp' = \pp$ and $\nn' = \nn$. Since $\pp$ is a vertex of
  $\text{newton}(f)$ with respect to $\nn$, $\nn^T\pp > \nn^T\pp_1$ for all
  $\pp_1 \in \Frame(f) \setminus \{\pp\}$. This implies that
  $\nn^T\pp > \nn^T\pp_1$ for all
  $\pp_1 \in \Frame^-(f) \setminus \{\pp\} = \left(\Frame^-(f) \cup
    \{\pp\}\right) \setminus \{\pp\}$. In other words, $\pp$ is a vertex of
  $\co(\Frame^-(f) \cup \{\pp\})$ with respect to $\nn$.

  Assume (ii). Suppose $V = \vx(\newton(f)) \subseteq \Frame^-(f)$. Then,
  $\pp' = \sum_{\Ss \in V}t_\Ss\Ss$ where $t_\Ss \in [0, 1]$,
  $\sum_{\Ss \in V}t_\Ss=1$. It follows that
  \[ \nn'^T\pp' = \sum_{\Ss \in V}t_\Ss\nn'^T\Ss < \sum_{\Ss \in
      V}t_\Ss\nn'^T\pp' = \nn'^T\pp' \sum_{\Ss \in V}t_\Ss = \nn'^T\pp',\] which
  is a contradiction. As a result, there must be some $\pp \in \Frame^+(f)$
  which is a vertex of $\newton(f)$ with respect to some $\nn \in \R^d$. \qed
\end{proof}
Thus some $\pp \in \Frame^+(f)$ is a vertex of the Newton polytope of a
polynomial $f$ if and only if the following formula is satisfiable:
\begin{align*}
\stropsat(f, \tovec{n'}, c)
&\deq \bigvee_{\pp \in \Frame^+(f)}\isvertex\left(\pp, 
\Frame^-(f)\cup \{\pp\}, \tovec{n'}, c\right)\\
&\equiv \bigvee_{\pp \in \Frame^+(f)}
\left[\tovec{n'}^T\pp+c>0 \land 
\bigwedge_{\qq \in \Frame^-(f)}\tovec{n'}^T\qq+c<0\right]\\
 &\equiv \left[\bigvee_{\pp \in \Frame^+(f)} 
 \tovec{n'}^T\pp+c>0\right] \land \left[\bigwedge_{\pp \in 
 \Frame^-(f)}\tovec{n'}^T\pp+c<0\right].
\end{align*}

For the case of several polynomials, the following theorem is a direct
consequence of Lemma~\ref{lemma:gurobi-lra-eq}.

\begin{theorem}
  \label{theorem:multiple-polys}
  Polynomials $\{f_i\}_{i\in\{1,\dots,m\}}$ have a positive \vbasis 
  with respect to $\nn \in \mathbb{R}^d$ if and only if 
  $\bigwedge_{i=1}^m\stropsat(f_i, \nn, c_i)$ is satisfiable. \qed
\end{theorem}

The formula $\bigwedge_{i=1}^m\stropsat(f_i, \nn, c_i)$ can be checked for
satisfiability using combinations of linear programming techniques and DPLL($T$)
procedures~\cite{Dutertre2006,Ganzinger2004}, i.e., satisfiability modulo linear
arithmetic on reals. Any SMT solver supporting the QF\_LRA logic is suitable. In
the satisfiable case $\{f_i\}_{i\in\{1,\dots,m\}}$ has a positive \vbasis and we
can construct a solution for $\bigwedge_{i=1}^mf_i>0$ as discussed earlier.

\begin{example}\label{eg:final}
Consider 
$f_1 = -12 + 2x^{12}y^{25}z^{49}-31x^{13}y^{22}z^{110}-11x^{1000}y^{500}z^{89}$ 
and $f_2 = -23+5xy^{22}z^{110}-21x^{15}y^{20}z^{1000} + 2x^{100}y^{2}z^{49}$. 
With $\nn = (n_1, n_2, n_3)$ this yields
\begin{align*}
\stropsat(f_1, \nn, c_1) &= 12n_1+25n_2+49n_3+c_1>0 \land 
13n_1+22n_2+110n_3+c_1<0 \\ &\land 1000n_1+500n_2+89n_3+c_1<0 \land c_1 < 
0,\\
\stropsat(f_2, \nn, c_2) &= (n_1+22n_2+110n_3+c_2>0 \lor 
100n_1+2n_2+49n_3+c_2>0) \\ &\land 15n_1+20n_2+1000n_3+c_2<0 \land c_2 < 0.
\end{align*}
The conjunction $\stropsat(f_1, \nn, c_1)\land\stropsat(f_2, \nn, c_2)$ is
satisfiable. The SMT solver CVC4 computes
$\nn=(-\frac{238834}{120461}, \frac{2672460}{1325071},-\frac{368561}{1325071})$
and $c_1 = c_2 = -1$ as a model. Theorem~\ref{theorem:multiple-polys} and
Proposition~\ref{prop:multiple-inequalities} guarantee that there exists a large
enough $a \in \R^+$ such that $f_1(a^\nn) > 0 \land f_2(a^\nn) > 0$. Indeed,
$a=2$ already yields $f_1(a^\nn) \approx 16371.99$ and
$f_2(a^\nn) \approx 17707.27$. \qed
\end{example}

\section{More General Solutions}\label{SE:gensol}
%
%
%
%
\label{sec:sign-generalize}
So far all variables were assumed to be strictly positive, i.e., only solutions
$\xx \in\mathopen]0, \infty\mathclose[^d$ were considered. This section proposes
a method for searching over $\R^d$ by encoding sign conditions along with the
condition in Theorem~\ref{theorem:multiple-polys} as a quantifier-free formula
over linear real arithmetic.

Let $V = \{x_1, \dots, x_d\}$ be the set of variables. We define a \emph{sign
  variant} of $V$ as a function $\tau: V \mapsto V \cup \{-x \mid x \in V\}$
such that for each $x \in V$, $\tau(x) \in \{x, -x\}$.  We write $\tau(f)$ to
denote the substitution $f(\tau(x_1), \dots, \tau(x_d))$ of $\tau$ into a
polynomial $f$. Furthermore, $\tau(a)$ denotes $\bigl(\frac{\tau(x_1)}{x_1}a,\dots,
\frac{\tau(x_d)}{x_d}a\bigr)$ for $a \in \R$.  A sequence $(\pp_1, \dots, \pp_m)$ is
a \emph{variant positive \vbasis} of $\{f_i\}_{i\in\{1,\dots,m\}}$ with respect
to a vector $\nn \in \R^d$ and a sign variant $\tau$ if $(\pp_1, \dots, \pp_m)$
is a positive \vbasis of $\{\tau(f_i)\}_{i\in\{1,\dots,m\}}$. Note that the
substitution of $\tau$ into a polynomial $f$ does not change the exponent
vectors in $f$ in terms of their exponents values, but only possibly changes
signs of monomials.  Given $\pp = (p_1, \dots, p_d) \in \N^d$ and a sign variant
$\tau$, we define a formula $\isnegated(\pp, \tau)$ such that it is $\true$ if 
and only if the sign of the monomial associated with $\pp$ is changed after 
applying 
the substitution defined by $\tau$: 
\begin{displaymath}
\isnegated(\pp, \tau) \deq \bigoplus\limits_{i=1}^d \bigl(\tau(x_i)=-x_i \land (p_i 
\bmod 2 = 1)\bigr).
\end{displaymath}
Note that this xor expression becomes $\true$ if and only if an odd number of
its operands are $\true$. Furthermore, a variable can change the sign of a
monomial only when its exponent in that monomial is odd. As a result, if
$\isnegated(\pp, \tau)$ is \tomath{true}, then applying the substitution defined
by $\tau$ will change the sign of the monomial associated with $\pp$. In
conclusion, some $\pp \in \Frame(f)$ is in the positive frame of $\tau(f)$ if
and only if one of the following mutually exclusive conditions holds:
\begin{enumerate}[(i)]
\item $\pp \in \Frame^+(f)$ and $\isnegated(\pp, \tau) = \false$
\item $\pp \in \Frame^-(f)$ and $\isnegated(\pp, \tau) = \true$.
\end{enumerate}
In other words, $\pp$ is in the positive frame of $\tau(f)$ if and only if the
formula
$\ispos(\pp, f, \tau) \deq \big(f_\pp>0 \land \neg\isnegated(\pp, \tau)\big)
\lor \big(f_\pp < 0 \land \isnegated(\pp, \tau)\big)$ holds. Then, the positive
and negative frames of $\tau(f)$ parameterized by $\tau$ are defined as
\begin{eqnarray*}
\Frame^+(\tau(f)) &=& \{\, \pp \in \Frame(f) \mid \ispos(\pp,f,\tau)\, \},\\
\Frame^-(\tau(f)) &=& \{\, \pp \in \Frame(f) \mid \neg\ispos(\pp,f,\tau)\, \},
\end{eqnarray*}
respectively.
The next lemma provides a sufficient condition for the existence of a solution
in $\R^d$ of $\bigwedge_{i=1}^mf_i>0$.

\begin{lemma}
        \label{lem:sign-variant}
 If there exists a variant positive \vbasis of 
 $\{f_i\}_{i\in\{1,\dots,m\}}$ with 
 respect to $\nn \in \R^d$ and a sign variant $\tau$, then 
 there exists $a_0 \in \R^+$ such that for all $a \in \R^+$ with $a \ge a_0$ 
 the 
 following holds:
 \begin{displaymath}
   \bigwedge_{i=1}^mf_i\big(\tau(a)^\nn\big)>0.
\end{displaymath}
\end{lemma}

\begin{proof}
Since  
$\{\tau(f_i)\}_{i\in\{1,\dots,m\}}$ has a positive \vbasis with respect 
to $\nn$, 
Proposition~\ref{prop:multiple-inequalities} guarantees that there exists $a_0 
\in \mathbb{R}$ such that for all $a \in \R$ with $a \ge a_0$, we have 
$\bigwedge_{i=1}^m\tau(f_i)(a^\nn)>0$, 
or $\bigwedge_{i=1}^mf_i\big(\tau(a)^\nn\big)>0$.
\qed
\end{proof}
A variant positive \vbasis exists if and only if there exist $\nn \in 
\R^d$, $c_1, \dots, c_m \in \R$, and a sign variant $\tau$ such that the 
following formula becomes $\true$:
 \begin{displaymath}
\generalstropsat(f_1, \dots, f_m, \nn, 
 c_1, \dots, c_m, \tau)
 \deq\bigwedge_{i=1}^m\stropsat\big(\tau(f_i),\nn,c_i\big),
 \end{displaymath}
where  for $i\in\{1,\dots, m\}$:
\begin{eqnarray*}
\stropsat\big(\tau(f_i),\nn,c_i\big) 
&\equiv &
\left[\bigvee_{\pp \in \Frame^+(\tau(f_i))}\hspace*{-3pt} \nn^T\pp+c_i>0\right] \land 
\left[\bigwedge_{\pp \in \Frame^-(\tau(f_i))}\hspace*{-3pt} \nn^T\pp+c_i<0 \right]\\
&\equiv & \phantom{\land\ }
\left[ \bigvee_{\pp \in \Frame(f_i)}
 \ispos(\pp,f_i, \tau) \land \nn^T\pp+c_i>0\right]\\ 
& & \land\ \left[
    \bigwedge_{\pp \in \Frame(f_i)} \ispos(\pp,f_i, \tau) \lor
    \nn^T\pp+c_i<0 \right].
\end{eqnarray*}

The sign variant $\tau$ can be encoded as $d$ Boolean variables $b_1$,
\dots,~$b_d$ such that $b_i$ is $\true$ if and only if $\tau(x_i) = -x_i$ for
all $i \in \{1, \dots, d\}$. Then, the formula
$\generalstropsat(f_1,\dots,f_m,\nn,c_1,\dots,c_m,\tau)$ can be checked for
satisfiability using an SMT solver for quantifier-free logic with linear real
arithmetic.

\section{Application to SMT Benchmarks}\label{SE:experiments}
A library STROPSAT implementing Subtropical Satisfiability, is available on our
web page\footnote{\url{http://www.jaist.ac.jp/~s1520002/STROPSAT/}}. It is
integrated into veriT~\cite{veriT} as an incomplete theory solver for non-linear
arithmetic benchmarks. We experimented on the QF\_NRA category of the SMT-LIB on
all benchmarks consisting of only inequalities, that is $4917$ formulas out of
$11601$ in the whole category. The experiments thus focus on those $4917$
benchmarks, comprising $3265$ \tomath{sat}-annotated ones, $106$
\tomath{unknown}s, and $1546$ \tomath{unsat} benchmarks. We used the SMT solver
CVC4 to handle the generated linear real arithmetic formulas
$\generalstropsat(f_1,\dots,f_m,\nn,c_1,\dots,c_m,\tau)$, and we ran veriT (with
STROPSAT as the theory solver) against the clear winner of the SMT-COMP 2016 on
the QF\_NRA category, i.e., Z3 (implementing nlsat~\cite{nlsat}), on a CX250
Cluster with Intel Xeon E5-2680v2 2.80GHz CPUs. Each pair of benchmark and
solver was run on one CPU with a timeout of 2500 seconds and 20 GB memory.
\tung{The experimental data and the library are also available on
  Zenodo\footnote{\url{http://doi.org/10.5281/zenodo.817615}}.}

Since our method focuses on showing satisfiability,
only brief statistics on \tomath{unsat} benchmarks are provided. Among the
$1546$ \tomath{unsat} benchmarks, 200 benchmarks are found unsatisfiable already
by the linear arithmetic theory reasoning in veriT. For each of the remaining
ones, the method quickly returns \tomath{unknown} within 0.002 to 0.096 seconds,
with a total cumulative time of 18.45 seconds (0.014 seconds on average). This
clearly shows that the method can be applied with a very small overhead, upfront
of another, complete or less incomplete procedure to check for unsatisfiability.

Table~\ref{tab:exp-result} provides the experimental results on benchmarks with
\tomath{sat} or \tomath{unknown} status, and the cumulative times. The
meti-tarski family consists of small benchmarks (most of them contain 3 to 4
variables and 1 to 23 polynomials with degrees between 1 and 4). Those are proof
obligations extracted from the MetiTarski project~\cite{Akbarpour2010}, where
the polynomials represent approximations of elementary real functions; all of
them have defined statuses. The zankl family consists of large benchmarks (large
numbers of variables and polynomials but small degrees) stemming from
termination proofs for term-rewriting systems~\cite{Fuhs2007}.
\begin{table}[hbt]
        \centering
        \caption{Comparison between STROPSAT and Z3 (times in seconds)}
        \label{tab:exp-result}
        \begin{tabular}{@{}l| r| r| r |r |r |r|r|r@{}}
                \hline
                \multirow{2}{*}{Family}& \multicolumn{4}{c|}{STROPSAT} & 
                \multicolumn{4}{c}{Z3} 
\\              \cline{2-9}
                & \tomath{sat} & Time & \tomath{unkown} & Time & 
                \tomath{sat} & Time & \tomath{unsat} & Time
                \\\hline
                meti-tarski (\tomath{sat} - 3220)& 2359 & 32.37 & 861 & 
                10.22 &  \bf 3220 & 88.55 &0 &0
                \\\hline
                zankl (\tomath{sat} - 45) & 29 & 3.77 & 16 & 0.59 & \bf 
                42 & 2974.35 &0 &0 \\\hline
                zankl (\tomath{unknown} - 106) & \bf 15 & 2859.44 & 76 & 
                6291.33 &
                14 & 1713.16  &23&1.06\\\hline
        \end{tabular}
        
\end{table}

        

Although Z3 clearly outperforms STROPSAT in the number of solved benchmarks, the
results also clearly show that our method is a useful complementing heuristic
with little drawback, to be used either upfront or in portfolio with other
approaches. As already said, it returns \tomath{unknown} quickly on
\tomath{unsat} benchmarks. In particular, on all benchmarks solved by Z3 only,
STROPSAT returns \tomath{unknown} quickly (see
Fig.~\ref{fig:ex-unknown-sat_unsat_timedout}).

When both solvers can solve the same benchmark, the running time of STROPSAT is
comparable with Z3 (Fig.~\ref{fig:ex-sat-timedout}). There are $11$ large
benchmarks ($9$ of them have the \tomath{unknown} status) that are solved by
STROPSAT but time out with Z3. STROPSAT times out for only $15$ problems, on
which Z3 times out as well. STROPSAT provides a model for $15$ \tomath{unknown}
benchmarks, whereas Z3 times out on 9 of them. The virtual best solver (i.e.\
running Z3 and STROPSAT in parallel and using the quickest answer) decreases the
execution time for the meti-tarski problems to 54.43 seconds, solves all
satisfiable zankl problems in 1120 seconds, and 24 of the unknown ones in 4502
seconds.

Since the exponents of the polynomials become coefficients in the linear
formulas, high degrees do not hurt our method significantly. As the SMT-LIB does
not currently contain any inequality benchmarks with high degrees, our
experimental results above do not demonstrate this claim. However, formulas like
in Example~\ref{eg:final} are totally within reach of our method (STROPSAT
returned \tomath{sat} within a second) while Z3 runs out of memory ($20$ GB)
after 30 seconds for the constraint $f_1>0 \land f_2 > 0$.

\begin{figure}[t]
        \centering
    \input{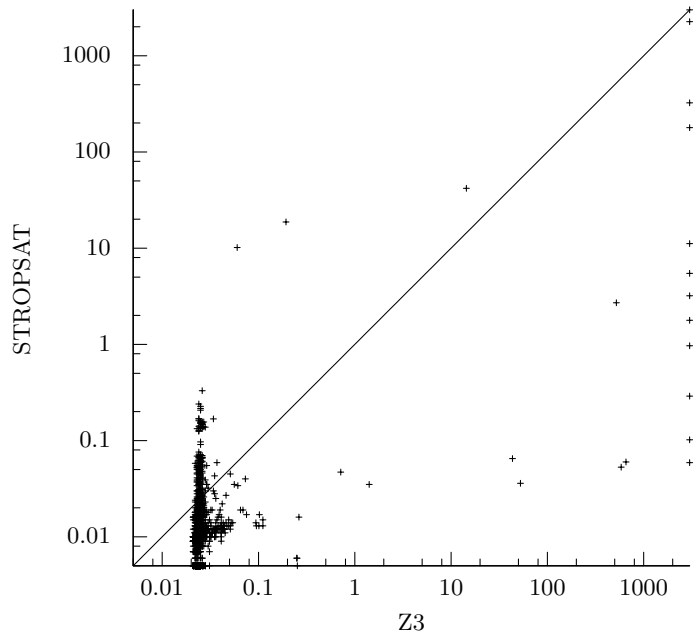}
        \caption{ STROPSAT returns \tomath{sat} or timeout ($2418$ benchmarks, times in seconds)}
        \label{fig:ex-sat-timedout}
\end{figure}

\begin{figure}[t]
        \centering
    \input{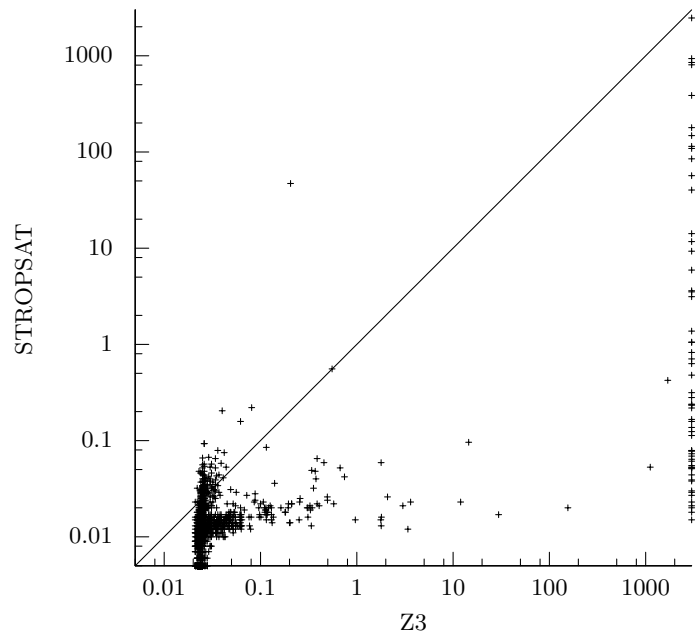}
    \caption{STROPSAT returns \tomath{unknown} ($2299$ benchmarks, times in seconds)}
    \label{fig:ex-unknown-sat_unsat_timedout}
\end{figure}

\section{Conclusion}\label{SE:conclusions}

We presented some extensions of a heuristic method to find simultaneous positive
values of nonlinear multivariate polynomials. Our techniques turn out useful to
handle SMT problems. In practice, our method is fast, either to succeed or to
fail, and it succeeds where state-of-the-art solvers do not. Therefore it
establishes a valuable heuristic to apply either before or in parallel with
other more complete methods to deal with non-linear constraints. Since the
heuristic translates a conjunction of non-linear constraints one to one into a
conjunction of linear constraints, it can easily be made incremental by using an
incremental linear solver.

To improve the completeness of the method, it could be helpful to not only
consider vertices of Newton polytopes, but also faces. Then, the value of the
coefficients and not only their sign would matter. Consider
$\{\pp_1, \pp_2, \pp_3\} = \face(\nn, \newton(f))$, then we have
$\nn^T\pp_1 = \nn^T\pp_2 = \nn^T\pp_3$. It is easy to see that
$f_{\pp_1}\xx^{\pp_1} + f_{\pp_2}\xx^{\pp_2} + f_{\pp_3}\xx^{\pp_3}$ will
dominate the other monomials in the direction of $\nn$. In other words, there
exists $a_0 \in \R$ such that for all $a \in \R$ with $a \ge a_0$,
$\sign(f(a^\nn)) = \sign(f_{\pp_1}+f_{\pp_2}+f_{\pp_3})$. We leave for future
work the encoding of the condition for the existence of such a face into linear
formulas.

In the last paragraph of Section~\ref{SE:subtropical}, we showed that, for the
subtropical method to succeed, the set of values for which the considered
polynomial is positive should either be unbounded, or should contain points
arbitrarily near $\zero$.  We believe there is a stronger, sufficient condition,
that would bring another insight to the subtropical method.

We leave for further work two interesting questions suggested by a reviewer,
both concerning the case when the method is not able to assert the
satisfiability of a set of literals.  First, the technique could indeed be used
to select, using the convex hull of the frame, some constraints most likely to
be part of an unsatisfiable set; this could be used to simplify the work of the
decision procedure to check unsatisfiability afterwards.  Second, a careful
analysis of the frame can provide information to remove some constraints in
order to have a provable satisfiable set of constraints; this could be of some
use for in a context of max-SMT.

Finally, on a more practical side, we would like to investigate the use of the
techniques presented here for the testing phase of the raSAT
loop~\cite{raSAT-ijcar2016}, an extension the interval constraint propagation
with testing and the Intermediate Value Theorem. We believe that this could lead
to significant improvements in the solver, where testing is currently random.


\section*{Acknowledgments}
We are grateful to the anonymous reviewers for their comments. This research has
been partially supported by the ANR/DFG project SMArT (ANR-13-IS02-0001 \& STU
483/2-1) and by the European Union project SC\textsuperscript{2}
(
grant agreement No.\ 712689).
The work has also received funding from the European Research Council under the
European Union's Horizon 2020 research and innovation program (grant agreement
No.\ 713999, Matryoshka). The last author would like to acknowledge the JAIST
Off-Campus Research Grant for fully supporting him during his stay at LORIA,
Nancy. The work has also been partially supported by the JSPS KAKENHI
Grant-in-Aid for Scientific Research(B) (15H02684) and the JSPS Core-to-Core
Program (A. Advanced Research Networks).

\bibliographystyle{splncs_srt}
\bibliography{generic}

\end{document}